\documentclass[a4paper]{llncs}

\begingroup
  \def\x{\endgroup\ExecuteOptions{dvipdfm}}%
  \ifx\pdfoutput\undefined
  \else
    \ifx\pdfoutput\relax
    \else
      \ifcase\pdfoutput
        
      \else
        \def\x{\endgroup\ExecuteOptions{pdftex}}%
	
      \fi
    \fi
  \fi
\x

\usepackage{tikz}  

\usepackage{graphicx}
\usepackage{amsfonts}
\usepackage{amsmath}
\usepackage{amssymb}
\usepackage{mathtools}
\usepackage{longtable}
\usepackage{ulem}
\usepackage{booktabs}
\usepackage{url}
\usepackage{ulem}
\usepackage[T1]{fontenc}
\usepackage{times}
\usepackage[bookmarks=true,colorlinks=true]{hyperref}    
\usepackage{array}
\usepackage[boxed,vlined,linesnumbered,noend]{algorithm2e}
\usepackage{tikz}
\usetikzlibrary{matrix,arrows,decorations.pathmorphing, ,decorations.pathreplacing, calc, shapes,patterns, fit, positioning}
\usepackage{float}
\newtheorem{fact}{Fact}
\SetKwInOut{Input}{Input}
\SetKwInOut{RestInput}{}
\SetKwInOut{Output}{Output}
\SetKwInOut{RestOutput}{}
\SetKwInOut{Actors}{Actors}
\SetAlFnt{\small}
\newcommand{\keywords}[1]{\par\addvspace\baselineskip
\noindent\keywordname\enspace\ignorespaces#1}



\newcommand{\papertitle}{Practical Fault-Tolerant Data Aggregation}

\newcommand{\paperthanks}{The study is cofounded by the European Union from resources of the European Social Fund. Project PO KL ,,Information technologies: Research and their interdisciplinary applications'', Agreement UDA-POKL.04.01.01-00-051/10-00.
(Third author). Contribution of the second author is supported by Polish National Science Center - DEC 2013/09/B/ST6/02258}

\makeatletter
\newlength\min@xx
\newcommand*\xxrightarrow[1]{\begingroup
  \settowidth\min@xx{$\m@th\scriptstyle#1$}
  \@xxrightarrow}
\newcommand*\@xxrightarrow[2][]{
  \sbox8{$\m@th\scriptstyle#1$}  
  \ifdim\wd8>\min@xx \min@xx=\wd8 \fi
  \sbox8{$\m@th\scriptstyle#2$} 
  \ifdim\wd8>\min@xx \min@xx=\wd8 \fi
  \xrightarrow[{\mathmakebox[\min@xx]{\scriptstyle#1}}]
    {\mathmakebox[\min@xx]{\scriptstyle#2}}
  \endgroup}
\newcommand*\xxleftarrow[1]{\begingroup
  \settowidth\min@xx{$\m@th\scriptstyle#1$}
  \@xxleftarrow}
\newcommand*\@xxleftarrow[2][]{
  \sbox8{$\m@th\scriptstyle#1$}  
  \ifdim\wd8>\min@xx \min@xx=\wd8 \fi
  \sbox8{$\m@th\scriptstyle#2$} 
  \ifdim\wd8>\min@xx \min@xx=\wd8 \fi
  \xleftarrow[{\mathmakebox[\min@xx]{\scriptstyle#1}}]
    {\mathmakebox[\min@xx]{\scriptstyle#2}}
  \endgroup}

\newcommand{\sect}[1]{\mbox{Sect.\,\ref{#1}}}  
\newcommand{\subsect}[1]{\mbox{Subsect.\,\ref{#1}}}  

\usepackage{xcolor}

\DeclareMathOperator{\bigO}{O}

\newcommand{\BigO}[1]{\bigO\left(#1\right)}

\newcommand{\Node}{\mathrm{\mathbf{N}}}
\newcommand{\Agg}{\mathrm{\mathbf{AGG}}}
\newcommand{\lAgg}{\mathrm{\mathbf{Agg}}}

\newcommand{\Enc}[2]{\mathrm{Enc_{#1}\left( #2\right)}}

\newcommand{\sk}{\mathrm{sk}}

\newcommand{\anonymous}[2]{#1} 

\newcommand{\compressaffil}[2]{#1} 

\begin{document}

\title{\papertitle}
\titlerunning{\papertitle}

\author{%
\anonymous{
\mbox{Krzysztof Grining}, \mbox{Marek Klonowski},  
\mbox{Piotr Syga}\thanks{\paperthanks}
\compressaffil{}{}
}
{Of anonymous authors}
}
\authorrunning{
\anonymous{K. Grining, M.~Klonowski, P.~Syga
}{Anonymous authors}}

\institute{
\anonymous{ \compressaffil{Faculty of Fundamental Problems of Technology,}{}
 Wroc{\l}aw University of Technology\\
 \compressaffil{\email{%
 \{firstname.secondname\}@pwr.edu.pl} }{}
}{...
}}

\maketitle

\begin{abstract}
 During Financial Cryptography 2012  Chan et al. presented a novel privacy-protection fault-tolerant data aggregation protocol. 
 Comparing to previous work, their scheme guaranteed provable privacy of individuals and could work even if some 
 number of users refused to participate. 

In our paper we demonstrate that despite its merits, their method provides unacceptably low accuracy of aggregated data for a wide range of 
assumed parameters and cannot be  used in majority of real-life systems. To show this we use both  analytic and experimental methods. 
 
 Additionally, we present a precise data aggregation protocol that provides provable level of security even when facing 
massive failures of nodes.  Moreover, the protocol requires significantly less computation 
(limited  exploiting of heavy cryptography) than most of currently known fault tolerant aggregation protocols and offers better security guarantees that make it suitable for systems of limited resources (including sensor networks).  To obtain our result we relax however the model and allow some limited communication between the nodes.

\keywords{Data aggregation,  differential privacy,  fault tolerance}
\end{abstract}

 \section{Introduction}\label{sect:intro}

Aggregation of data is a fundamental problem that has been approached from different perspectives. Recently there were  
many papers published, that presented methods of data aggregation that preserve privacy of individual users. 
More precisely, the goal of the protocol is to reveal some general aggregated statistics (like an average value) 
while keeping value of each individual secret, even if the aggregator is untrusted (e.g., tries to 
learn input of individual users). The general notion is to design a protocol that allows the aggregator to learn a perturbed sum, but no intermediate results.

In~\cite{PaniShi} Shi et al. have introduced a new approach to aggregation of information 
in distributed systems based on combining cryptographic techniques and typical ``methods of differential privacy''
, that was originally used for  protecting privacy of individuals in statistical data bases after  some data was revealed. 
The privacy preservation is usually realized by adding some carefully prepared noise to the aggregated values. Similar approach has been independently proposed in~\cite{Rastogi}.

Those papers put a new light on the problem of privacy preserving data aggregation -- the authors constructed a protocol that 
can be very useful, however its applicability is limited to some narrow class of scenarios due to few shortcomings. 
One of them  is the fact that \textbf{all}  of the members of a group of users have to 
cooperate to compute the aggregated data. Thus, this approach is not  appropriate  for a dynamic, real-life systems (e.g. mobile sensor networks),
even though it seems to be a perfect solution for fixed, small system of devices, where a series of data is generated periodically for a long time and the number of failures is always small (e.g. collecting measurements of electricity consumption in a neighborhood). \\
Another important  protocol, called Binary Protocol, has been introduced  in~\cite{Hubercik}, wherein authors presented the first privacy preserving aggregation protocol that  is, to some extent, fault tolerant. 
In our paper we focus on showing some shortcomings of the solution from~\cite{Hubercik} (by pointing out the extent to which it is fault tolerant)
as well as present our approach to privacy preserving and  fault tolerant data aggregation.

\subsection{Our Contribution and Organization of the Paper }\label{ssect:org}

In~\subsect{ssect:notation} we briefly describe the model 
assumed in our paper and provide some notation used throughout it as well as introduce some definitions we use further on. 
In~\sect{sect:BP} we recall the Binary Protocol by Chan et al. presented in~\cite{Hubercik}, 
followed by discussion of its disadvantages in~\sect{sect:diss}. 
In~\sect{sect:prot} we present and analyze our protocol addressing some of the Binary Protocol's issues. 
\sect{sect:prev} is devoted to recalling some of the previous work related to the 
problem addressed in the paper. Finally, in~\sect{sect:conclusion} we conclude and indicate some possible future work.
The contribution of our paper is twofold. 

\begin{itemize}
 \item We show that the fault tolerant protocol from~\cite{Hubercik} (called Binary Protocol)  
 offers very low level of accuracy of aggregated data even for small number of faults for any reasonable 
 size of the network. This holds despite very good asymptotic guarantees.  

 \item On the positive side we construct a modified protocol that offers much better accuracy and significantly 
lower computational requirements. We assume however a weaker security model where users may trust a few others
and we allow some limited, local communication between users. 
This assumption is justified in various scenarios, specifically when users have some \textbf{local} knowledge about few other participants. 
This is a natural assumption in electricity meters, where privacy concerns is that the adversary can deduce i.e. the sleep/work habits 
or the number of inhabitants in the household. Your neighbors knows your habits anyway.  Similarly, in cloud services or social network, 
where you naturally have some friends or users to whom you give your data on your own free will. More precisely, 
\textbf{all} my neighbors/friends can beak my privacy cooperating easier then any outer party.   
\end{itemize}


 \section{ Definitions and Tools}\label{ssect:notation}

Below we present some definitions and facts that will be used throughout this paper. We will denote the set of real numbers by $\mathbb{R}$, integers by $\mathbb{Z}$ and natural numbers by $\mathbb{N}$.

\begin{definition} (Symmetric Geometric Distribution). Let $\alpha > 1$.
We denote by $Geom(\alpha)$ the symmetric geometric distribution that takes integer values such that the probability mass
function at $k \in \mathbb{Z}$ is $\frac{\alpha -1 }{\alpha +1} \cdot \alpha^{-|k|}$. 
\end{definition}

\begin{fact}\label{fuckt2} (From \cite{Hubercik}) Let $\epsilon>0$. Let $u,v$ be integers such that $|u-v| \leq \Delta$ for fixed $\Delta \in \mathbb{N^+}$ . Let $r$ be a random variable having distribution 
 $Geom(\exp(\frac{\epsilon}{\Delta}))$. Then for any integer $k$ 

$$Pr[v+r=k] \leq \exp(\epsilon)\Pr[u+r=k].$$
\end{fact}

\begin{definition}\label{DGD} (Diluted Geometric Distribution). Let $\alpha > 1$ and $0< \beta \leq  1$.
A random variable has $\beta$-diluted Geometric distribution $Geom^{\beta}(\alpha)$ if with 
probability $\beta$ it is sampled from $Geom(\alpha)$, and  with probability $1-\beta$ is set to $0$.
\end{definition}

In the same manner as in~\cite{Hubercik}, we use \textit{computational differential privacy} as a measure of privacy protection. This notion has been introduced (in a similar form) in~\cite{CDP} and is in fact a computational counterpart  of differential privacy from~\cite{dW1}.

\begin{definition} (Computational Differential Privacy Against Compromise (from~\cite{Hubercik}))
Suppose the users are compromised by some underlying randomized process $\mathcal{C}$, and we use $C$
to denote the information obtained by the adversary from the compromised users. Let
$\varepsilon,\, \delta >0$. A (randomized) protocol $\Pi$ preserves computational $(\varepsilon,\, \delta)$-differential privacy
(against the compromising process $\mathcal{C}$) if there exists a negligible function $\eta : \mathbb{N}\rightarrow \mathbb{R}^+$
such that for all $\lambda\in \mathbb{N}$, for all $i \in \left\{1,\,2,\,\dots,\,n\right\}$, for all vectors $x,\,y\in\left\{0,\,1\right\}^n$ that differ only at position i, for all probabilistic polynomial-time Turing machines $\mathcal{A}$, for any output 
$b\in\left\{0,\,1\right\}$, \\
$$\Pr_{\mathcal{C}_i}\left[\mathcal{A}\left(\Pi\left(\lambda,\,x\right),\,C\right)=b\right]\leq e^{\varepsilon}\Pr_{\mathcal{C}_i}\left[\mathcal{A}\left(\Pi\left(\lambda,\,x\right),\,C\right)=b\right]+\delta+\eta\left(\lambda\right)~,$$
where the probability is taken over the randomness of $\mathcal{A}$, $\Pi$ and $\mathcal{C}_i$, which denotes the
underlying compromising process conditioning on the event that user $i$ is uncompromised.
\end{definition}
In a similar manner to regular differential privacy, we say that protocol $\Pi$ preserves computational $\varepsilon$-differential privacy if it preserves computational 
$(\varepsilon,\,0)$-differential privacy.
The intuition behind this definition is as follows. Every party has some bit $b$. From observing some processing of data, it is not feasible for any computationally bounded 
adversary to  learn too much about $b$. This should hold with probability at least $1-\delta$.

 \section{Protocol by Chan et al. -- Description}\label{sect:BP}

In the paper~\cite{Hubercik} authors propose a fault tolerant, privacy preserving data aggregation protocol which they named Binary Protocol. 
The purpose of the protocol is to allow some untrusted Aggregator $\Agg$, to learn the sum of values $x_i$, $1\leq i\leq n$, 
where where $x_i$ is kept by the $i$-th user. We will denote $i$-th user by $\Node_i$. The idea is based on earlier work~\cite{PaniShi}, in particular the Block Aggregation protocol. 
In this setting, we do not have a trusted party who can collect the data and perform some specific actions to preserve privacy (i.e. add noise of appropriate magnitude). The users themselves have to be responsible for securing their privacy by adding noise from some specific distribution, encrypting the noisy value and sending it to the Aggregator. This problem requires combination of both cryptographic and privacy preserving techniques. See that we have essentially two adversaries here. First is an external one, against whom we have to use cryptography to protect the communication between users and the Aggregator. This external adversary should not be able to decipher anything, including noisy sum of all data. On the other hand, the Aggregator himself is an adversary as well. This adversary, however, should be able to decrypt only the noisy sum (not the single user noisy values) and should not be able to compromise the privacy of any single user. 
The general notion behind Block Aggregation is to generate a random secret key $sk_i$ for each of $n$ users as well as an additional $sk_0$ given to the Aggregator, 
such that $\sum_{i=0}^n{sk_i}=0$. 
Before sending the encrypted data, $i$-th user  adds noise $r_i$ coming from Diluted Geometric Distribution (Def.~\ref{DGD} in Sect.~\ref{ssect:notation}). We will denote the noisy data of $i$-th user by $\tilde{x}_i = x_i + r_i$.
Namely, each user transmits  $\Enc{\mathrm{sk}_i}{\tilde{x}_i}$ so that upon receiving all shares and having $sk_0$, the secret keys cancel out and the Aggregator is left with the desired noisy sum. One may easily note that as long as each user transmits its value, $\Agg$ may use $sk_0$ to decipher the sum. 
The symmetric geometric distribution $Geom(\alpha)$ can be viewed as a discrete version of Laplace distribution, which is widely used in differential privacy papers. Having discrete values is essential for the cryptography part of the protocol. The dilution parameter $\beta$ is the probability that a specific user will add noise from $Geom(\alpha)$. This is done because, intuitively, we want at least one user to add a geometric noise, but we do not want too many of these noises to keep the necessary noise sufficiently small.
The problem that occurred with so-called Block Aggregation is that whenever a single user fails to deliver their share (and what is really important -- their $sk_i$), the blindings do not cancel out, hence making it impossible for the Aggregator to decipher the desired value.

Binary Protocol presented in~\cite{Hubercik} addresses the incompleteness of the data by arranging the users in a virtual binary tree. One may visualize each user as a leaf of a binary tree, with all the tree-nodes up to the root being virtual. The Aggregator is identified with an additional tree-node, which is located ``above'' the root and is connected only to the tree-root. In order to simulate the tree structure, the users and $\Agg$ are equipped with appropriate secret keys and generate random noises for each of the tree-layer, where layer is equivalent to the depth the tree-node is at, i.e., the first layer consists of root, second layer consists of two direct children of the root, and so on. Finally, at  the $\left\lceil \log n\right\rceil+1^{\text{st}}$ layer consists of the leaves. Finally, each user performs Block Aggregation protocol for each of the layers, i.e., they generate their block $\Enc{\mathrm{sk}_i}{\tilde{x}_i}$ for the $\left\lceil \log n\right\rceil+1^{\text{st}}$ 
layer and their shares for larger blocks of 
higher layers. In each of the layers, the noise $r_i$ is taken from a different distribution, namely $\beta$ parameter for diluted geometric distribution is derived as follows: $\beta=\min\left(\frac{1}{\left|B\right|}\ln \frac{1}{\delta_0},\,1\right)$, where $\left|B\right|$ is the number of tree-nodes in the layer and $\delta_0>0$ is a privacy parameter. One may note that, the more tree-nodes in the layer, the blinding becomes sparser.
If all users present their shares the problem is reduced to the original Block Aggregation. Namely, the Aggregator may decrypt the root-layer block, obtaining the sum 
of all the $\tilde{x}_i$ with the blinding canceled out. However, if at least one user $\Node_i$
fails, all the blocks containing $\Node_i$ will suffer the same issues as Block Aggregation with a missing user. Namely, large, 
uncanceled random disturbance. In order to provide the aggregation of the working users, the authors allow the Aggregator to find such a covering of the tree from the blocks of 
different layers that all the working users are covered, none of the failed users is included and that $\Agg$ is able to recover the result.

Binary Protocol provides security under computational differential privacy model and results in $\BigO{n\log n}$ communications exchanged in the network and guarantees 
$\tilde{\bigO}\left(\left(\log n\right)^{\frac{3}{2}}\right)$ error. This notion  hides significant constants. Nevertheless in a practical setting, those results are 
less satisfying than one would expect. The issues concerning the privacy and the resulting error are raised in~\sect{sect:diss}.

\section{Analysis of Chan et al.'s Protocol -- The Magnitude of Error}\label{sect:diss}
In this section we will show that the error magnitude in Binary Protocol is significant for moderate number of participants.
Note that in ~\cite{Hubercik} the authors assumed that each user has data $x_i \in \{0,1\}$, which means that the range of the sum of aggregated data is $[0,n]$.
Thus, error of magnitude $\gamma n$ shall be regarded large already for moderate constant  $\gamma$.
They have also shown that the magnitude of error is $o(n)$ asymptotically. However, in practical applications we are also interested in performance of this protocol for moderate values of $n$, i.e. $n \leqslant 2^{14}$. We will show that for a reasonable range of values of the number of users $n$ and number of failures $\kappa$ the error is large ($\gamma n$ for some  constant $\gamma$) with significant probability. Obviously, as the $n$ increases, the Binary Protocol becomes better because of the asymptotic guarantees. However, our aim here is to show, that if the number of participants is at most moderate (i.e. $2^{12}$) or the number of failures is significant (i.e. $\kappa = \log_2(n)$, $\kappa = \lfloor\frac{n}{2^6}\rfloor$) then the accuracy of Binary Protocol is too low to be used. Furthermore, if the number of users is quite small (i.e. $2^{10}$ or less), then even for $\kappa = 5$ the errors generated are unacceptably high.
\par
We aim to show a precise magnitude of error in the Binary Protocol. 
To achieve this, we will use some subtler method than these presented by the authors of ~\cite{Hubercik}. 
To support our analytic analysis we show results of simulations. Note that in ~\cite{Hubercik} the authors described only simulations without failures, 
even though their protocol is specifically designed to handle failed users. 

\subsection{Analytical Approach}{\label{ssect:expectedNoise}}
The size of error depends on the number of failed users and the way they are distributed amongst all participants. Let us fix $n$ as the number of participants. 
Like the authors of ~\cite{Hubercik}, we assume for simplicity that $n$ is a power of $2$. Our reasoning can be however generalized for every $n$. 
We also assume that $\kappa$ users have failed.
We assume that these failed users are uniformly distributed amongst all participants, which seems to be reasonable in most scenarios. 
The error generated during the Binary Protocol is the sum of all noises in the aggregated blocks. Throughout this section we will use following notation, $\delta_0 = \frac{\delta}{\lfloor\log_2(n)\rfloor + 1}$, where $\delta$ is a privacy parameter. Also we have $\beta_i=\min\left(\frac{1}{\left|B_i\right|}\ln \frac{1}{\delta_0},\,1\right)$, where $B_i$ is size of the node on $i$th level of the tree. Because we assumed that $n$ is a power of $2$, so the binary tree is full, then $B_i$ is essentially the number of leaves being descendants of any node on $i$th level of the tree.
In our analysis, first we show an exact formula for the expected value of the number of noises added by individual nodes. The exact formula is given in the following theorem.
\begin{theorem}{\label{ssect:thmEY}}
Let $Y$ be a random variable which denotes the number of noises added during the Binary Protocol. Let $\kappa > 0$ and fix $n$ as the number of participants. Then, the expected value of random variable $Y$ is given by the following formula:
$$
EY = n-\kappa + n\cdot \sum_{i=1}^{\log_2(n)-1}\left(\frac{\binom{n-\frac{n}{2^i}}{\kappa}}{\binom{n}{\kappa}} \cdot  \left(\beta_i - \beta_{i+1}\right)\right),
$$
where $\beta_i=\min\left(\frac{1}{\left|B_i\right|}\ln \frac{1}{\delta_0},\,1\right)$.
\end{theorem}
Proof of this theorem can be found in Appendix. It is based on  combinatorial and probabilistic techniques. Now we  show a lower bound for this value for limited range of $n$. We present it in the following
\begin{lemma}{\label{ssect:lemmaEY}}
Let $2^4 \leqslant n \leqslant 2^{21}$ and $\delta = 0.05$, then $EY$ has a following lower bound:
$$
EY \geqslant n-\kappa - n\cdot \left( e^{-\frac{8\kappa}{n}} + \frac{\ln(\frac{\log_2(n)+1}{\delta})}{8} \cdot \left( e^{-\frac{16\kappa}{n}} - e^{-\frac{8\kappa}{n}} \right)\right).
$$
\end{lemma}
Note that if $n < 2^4$ then we have $\beta_i = 0$, which means that every remaining user has to add noise (even if there are no failures, i.e $\kappa=0$). There is no need to give a lower bound in that case, because then the number of noisy inputs is exactly $n-\kappa$. Note also that even though we fixed a specific $\delta$ that is used broadly in previous papers  (including~\cite{Hubercik}), similar reasoning can be made for different values of $\delta$.   
\par
We can use this bound to obtain a following
\begin{corollary}{\label{ssect:corEY}}
Fix $\delta = 0.05$. For $n \leqslant 2^{10}$ and $\kappa = \log_2(n)$, we have $EY \geqslant 0.1n$. 
Similarly, if $\kappa = \lfloor\frac{n}{2^6}\rfloor$, then for $2^6 \leqslant n \leqslant 2^{12}$ we have $EY \geqslant 0.16n$.
\end{corollary}
This comes immediately from Lemma~\ref{ssect:lemmaEY} and an observation that $\frac{EY}{n}$ is a decreasing function of $n$. After plugging the greatest value of $n$ that is allowed by assumptions we obtain these bounds.
\par
Having an exact formula and also a lower bound for the expected number of noises generated, we can calculate the error. Let us assume that we have $m$ noises generated. 
Recall that each of them comes from symmetric geometric distribution $Geom(\alpha)$ with $\alpha > 1$, which is comprehensively described both in ~\cite{PaniShi} and ~\cite{Hubercik}. 
We denote the sum of all noises as $Z$. One can easily see that $EZ = 0$ due to symmetry of distribution. However the expected additional error i.e., $E|Z|$ might be, and we will show that it often is, quite large.

\begin{theorem}{\label{ssect:thmEZ}}
Consider Binary Protocol with fixed $\alpha$, let $m$ denote the number of noises generated, each coming from $Geom(\alpha)$ distribution. Then let $Z$ be a random variable which denotes the sum of generated noises. We have
$$
E|Z| = \int\displaylimits_{0}^{\infty} \frac{4\cdot \alpha \cdot m \cdot \sin{t} \cdot \left(\alpha-1\right)^{2m}}{t \cdot \pi \cdot \left(\alpha^2 - 2\alpha \cos{t} + 1\right)^{m+1}}dt.
$$

\end{theorem}
The proof of this theorem is presented in Appendix. It is based on techniques comprehensively described in~\cite{pinelis}. We also show a lower bound for $E|Z|$ in a following
\begin{lemma}{\label{ssect:lemmaEZ}}
For fixed $n$ and $\epsilon$, which is a privacy parameter, provided that $\alpha = \frac{\epsilon}{\log_2(n)+1}$ and $m = \gamma n$, for $\gamma \in [0,1]$ we have
$$
E|Z| \geqslant c_{n,\epsilon} \cdot \sqrt{\gamma} \cdot \frac{\log_2(n) \cdot \sqrt{n}}{\epsilon \sqrt{\pi}} - 0.1~,
$$ 
where $c_{n,\epsilon}$ is a constant, which is at least $1.4$ for moderate values of $n$ and $\epsilon$.
\end{lemma}
\noindent
Having all useful theorems and lemmas we can obtain a following
\begin{corollary}\label{ssect:corEZ}
Consider Binary Protocol for $\delta = 0.05$, $\epsilon = 0.5$, $n \leqslant 2^{10}$ and $\kappa = \log_2(n)$. 
Let $|Z|$ be the absolute value of all noises aggregated during this protocol. We have $E|Z| \geqslant 0.15\cdot n$. Moreover, if we take $\kappa = \frac{n}{2^6}$ and $2^6 \leqslant n \leqslant 2^{12}$ we have $E|Z| \geqslant 0.12 \cdot n$.
\end{corollary}
This is an immediate result from Lemma.~\ref{ssect:lemmaEZ}, we can see that $\frac{E|Z|}{n}$ is a decreasing function of $n$, so it is enough to plug $n = 2^{10}$ into lower bound for $E|Z|$ for the first part of the corollary and $n=2^{12}$ for the second part of the corollary.
\par
This clearly shows that even if we consider the lower bound for the number of noises and their magnitude, the Binary Protocol is far from perfect for many realistic scenarios, i.e. when the number of participants is moderate. Even worse conclusions will be drawn in~\subsect{ssect:experimental}, where we use the exact formulas given in Theorems \ref{ssect:thmEY} and \ref{ssect:thmEZ} to numerically analyze the errors generated in this protocol.
\subsection{Experimental Approach}{\label{ssect:experimental}}
In~\subsect{ssect:expectedNoise} we gave both exact formulas and lower bounds for the number of noises generated and their sum. Note that the lower bounds are not very tight for many $n$. In this subsection we will show that the errors generated are, in fact, even larger. We will use the exact formulas to precisely calculate the errors numerically. First let us consider the case where $n \leqslant 2^{10}$, $\kappa = \left \lfloor{\log_2(n)}\right \rfloor$, and privacy parameters are $\epsilon = 0.5$, $\delta = 0.05$. See Fig.~\ref{fig:err1}. It clearly shows that the error magnitude in Binary Protocol is, in fact, significantly greater than the lower bound given  in Corollary~\ref{ssect:corEZ}. Now let $2^6 \leqslant n \leqslant 2^{12}$, $\kappa = \frac{n}{2^6}$ and privacy parameters stays the same. See Fig.~\ref{fig:err2}. Again we can see that the error magnitude is unacceptably high, greater than $0.2n$. Note that the noise is independent from the data, so such error could be very problematic, especially if 
the sum of the real data is small (e.g $o(n)$). In such case the noise could be greater than the data itself.
We can also check how great the errors will be for constant value of $\kappa = 5$. See Fig.~\ref{fig:err3}.


\begin{figure}[h!]
    \centering
    \includegraphics[width=0.45\textwidth]{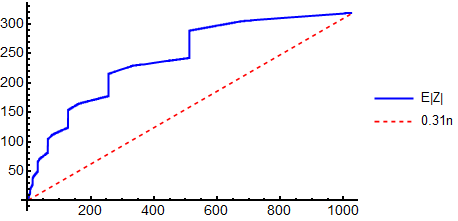}
    \caption{Error magnitude in Binary Protocol with $\epsilon = 0.5$, $\delta = 0.05$ and $\kappa = \left \lfloor{\log_2(n)}\right \rfloor$.}
    {\label{fig:err1}}
\end{figure}

\begin{figure}[h!]
    \centering
    \includegraphics[width=0.45\textwidth]{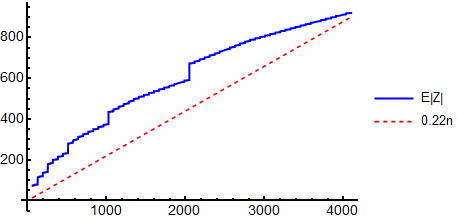}
    \caption{Error magnitude in Binary Protocol with $\epsilon = 0.5$, $\delta = 0.05$ and $\kappa = \left \lfloor{\frac{n}{2^6}}\right \rfloor$.}
    {\label{fig:err2}}
\end{figure}

\begin{figure}[h!]
    \centering
    \includegraphics[width=0.45\textwidth]{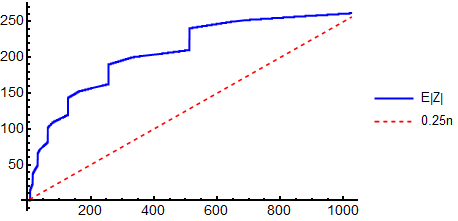}
    \caption{Error magnitude in Binary Protocol with $\epsilon = 0.5$, $\delta = 0.05$ and $\kappa = 5$.}
    {\label{fig:err3}}
\end{figure}

\subsection{Some Other Shortcomings}

Note that in~\cite{Hubercik}, but also in numerous other papers concerning data aggregation with untrusted aggregator, we have a constant privacy parameter $\delta$ (i.e. $\delta = 0.05$). This significantly decreases the amount of noises generated, but is fundamentally incorrect in terms of classic differential privacy standards. 
Such approach allows choosing $\delta$ fraction of the database and revealing their data to everyone. In fact, the magnitude of $\delta$ should be $o(\frac{1}{n})$, where $n$ is the number of users. This is necessary to ensure that the probability of leakage is negligible. More information about this can be found in~\cite{DworkAlgo}.
Furthermore, we assumed that no participants collude with the aggregator. We used the same regime as in~\cite{PaniShi} with $\gamma$ being a lower bound for fraction of non-colluding participants, the magnitude of added noises would be even greater.

 \section{Precise Aggregation Algorithm with Local Communication}\label{sect:prot}

In this part we present an alternative protocol PAALEC (Precise Aggregation Algorithm with Local Communication) that in some scenarios offers much better accuracy of aggregated data when failures occur, while preserving high level of users' privacy protection. 
In fact our protocol works in a substantially different way and for slightly modified model. Thus, despite its performance and accuracy that outperforms the original protocol of Chan at al., they are not fully comparable.

First of all, we assume that users may communicate (also in order to bypass the lower bound pointed out in \cite{HubercikBound}). 
Let us stress that the communication is limited to a small circle of ``neighbors''. The idea behind the presented construction is to take advantage of some natural structures emerging in distributed systems (e.g. social networks) wherein, apart from  logical connections between each user and a server/aggregator there are also some direct links between individual users.   
Clearly, such model is not adequate for some real-life problems discussed in~\cite{Hubercik}, for example in sensor fields with unidirectional communication. Thus there are applications where the original protocol from~\cite{Hubercik} is the only one possible.

\subsection{Modified Model}

We assume that the network consists of $n$ users - $V=\{v_1,v_2,\ldots , v_n\}$ as well as the aggregator $\Agg$ and a set of $k<n$ \textit{local aggregators} $\lAgg_1,\dots ,\, \lAgg_k$. Please note that the local aggregators may be separate entities but without any significant changes they may be selected from the set of regular users $V$. The only issue with this approach is that we have to ensure that the local aggregator is either selected during the aggregation round or it cannot fail during \textbf{a single} execution of aggregation process. We assume that each user is assigned to \textbf{exactly one} local aggregator. We denote the set of nodes assigned to the local aggregator $\lAgg_i$ by $V_i$. An example of the network's topology is depicted in Fig.~\ref{fig:topo}.

\tikzstyle{peers}=[draw,circle,black,bottom color=white,
                  top color= white, text=black,minimum width=10pt]
\tikzstyle{lAgg}=[draw,circle,gray,bottom color=gray, top color=gray,
                       text=white,minimum width=20pt]
\tikzstyle{AGG}=[rectangle, rounded corners, thin,
                           black, fill= black, draw, text=white,
                           minimum width=2.5cm, minimum height=0.8cm]
													
\begin{figure}
\centering
\begin{tikzpicture}[auto, thick,yscale=0.85]
   \node[AGG] (Agg) at (0,4) {\small{$\Agg$}};
  \foreach \place/\name in {{(-4.5,0.5)/a}, {(2,0)/b}, {(2,2)/c}, {(0,2)/d},
           {(-2,0)/e}}
    \node[lAgg] (\name) at \place {$\lAgg$};
		  \foreach \source in {a, b, c, d, e}
    \path (\source) edge (Agg);
		
   %
  \foreach \pos/\i in {above left of/1, left of/2, below left of/3, below of/4}
    \node[peers, \pos = e] (e\i) {};
   \foreach \speer/\peer in {e/e1,e/e2,e/e3, e/e4, e1/e2,e2/e4, e1/e3, e3/e4}
    \path (\speer) edge (\peer);
		
		\node[ellipse,draw=black,thick,dotted,fit=(e) (e1) (e2) (e3) (e4),inner sep=0pt] {};
		
   \foreach \pos/\i in {above right of/1, right of/2, below right of/3}
    \node[peers, \pos =b ] (b\i) {};
   \foreach \speer/\peer in {b/b1,b/b2,b/b3, b1/b2, b2/b3,b3/b1}
   \path (\speer) edge (\peer);
	
			\node[ellipse,draw=black,thick,dotted,fit=(b) (b1) (b2) (b3),inner sep=0pt] {};
	 \foreach \pos/\i in {below right of/1, below of/2}
   \node[peers, \pos =d ] (d\i) {};
   \foreach \speer/\peer in {d/d1,d/d2,d1/d2}
   \path (\speer) edge (\peer);
	
			\node[ellipse,draw=black,thick,dotted,fit=(d) (d1) (d2),inner sep=0pt] {};
   \foreach \pos/\i in {below left of/1, below of/2}
   \node[peers, \pos =a ] (a\i) {};
   \foreach \speer/\peer in {a/a1,a/a2, a2/a1}
   \path (\speer) edge (\peer);
	
			\node[ellipse,draw=black,thick,dotted,fit=(a) (a1) (a2),inner sep=0pt] {};
   \foreach \pos/\i in {above right of/1, right of/2, above of/3}
   \node[peers, \pos =c ] (c\i) {};
   \foreach \speer/\peer in {c/c1,c/c2, c/c3,c3/c1,c1/c2}
   \path (\speer) edge (\peer);
		\node[ellipse,draw=black,thick,dotted,fit=(c) (c1) (c2) (c3),inner sep=0pt] {};
	   \foreach \speer/\peer in {c2/b1,b2/c2,b1/d1,d2/e4,e2/a2,e3/a2}
   \path (\speer) edge (\peer);
	
\end{tikzpicture}
  \caption{Example of a clusterized network with global aggregator ($\Agg$) and local aggregators ($\lAgg$) marked.}\label{fig:topo}
\end{figure}
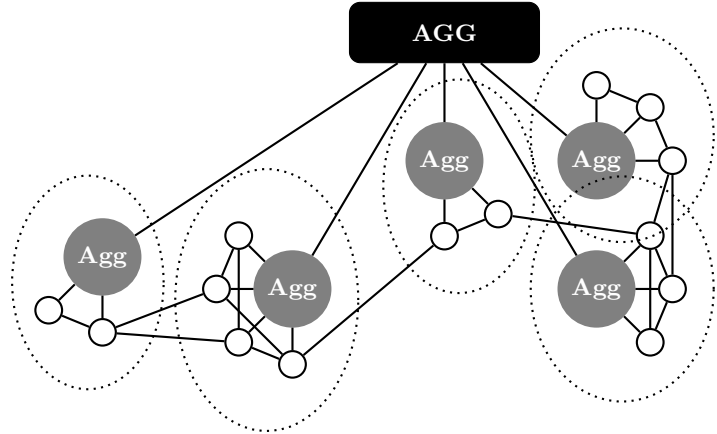 

We can derive a graph $G=(V,\,E)$ from the network structure, where $V$ are all the nodes and the set of edges is created based on the ability 
to establish communication (e.g., transmission range in a sensor network, friendship relation  in a social network).
Namely, the edge $\left\{v,\,v'\right\}\in E$ if and only if  $v$ and $v'$ are \textit{neighbors} and  can communicate via a private channel. 
In our protocol we assume that each node can perform some basic cryptographic operations and has access to a source of randomness.
By $N(v)$ we denote a set of such vertices $v'$ of $G$ that the edge $\left\{v,\,v'\right\}\in E$. 
Security of the protocol described in Section~\ref{ssect:desc}  depends on  the structure of graph $G$, and how many parties the adversary can corrupt. 
Discussion on security of the protocol can be found in Section~\ref{secdisc}.

\paragraph{Adversary.} The adversary may corrupt a subsets of users, local aggregators and the aggregator. 
It can read all messages the controlled parties sent or received.  The aim of the adversary is to learn about individual contributions of uncorrupted users.  

\subsection{Building Blocks}\label{ssect:bild}

Similarly to previous papers, for obtaining high level of data privacy  we combine cryptographic techniques with data perturbation methods typical for research concentrated on differential privacy of databases.  

The first technique we use in our protocol is a homomorphic encryption scheme based on original ElGamal construction enriched by some extra techniques introduced in~\cite{URE}. 
More precisely, encrypted messages can be ``aggregated'' and  re-encrypted. Moreover one can ``add'' an extra encryption layer to a given ciphertext, in such way that the message can be decrypted only using both respective keys. Clearly this operation preserves the homomorphic property. 

Let $\mathbf{G}$ be a group such that the Diffie-Hellman problem is hard. Let $g$ be a generator of $\mathbf{G}$. 
Let $\sk, \sk'$ be a some private keys and $g^{\sk}, g^{\sk'}$ are respective public keys. 

\begin{description}
 \item[Encryption of '$1$'] 
    A pair  $\Enc{\sk}{1}=(g^r,\,g^{r\cdot\sk})$ for a random $r\in \mathbf{G}$ is an encryption of $1$. 
 \item[Re-encryption] 
    Ciphertext representing $1$ can be re-encrypted. Namely, one can get another ciphertext representing one, \textbf{without private key}. 
    Namely having $\Enc{\sk}{1}=(g^r,\,g^{r\cdot\sk})$ one can choose $r'$ and compute $\mbox{Re}(\Enc{\sk}{1})=(g^{r\cdot r'},\,g^{{r\cdot r'}\cdot\sk})$ that represents $1$ as well.

 \item[Adding layer of encryption]
        Having  a ciphertext  $\Enc{\sk}{1}=(g^r,\,g^{r\cdot\sk})$ a party having private key $\sk'$  can ``add encryption layer''
      to a ciphertext obtaining  $$\Enc{\sk+\sk' }{1}=((g^r)^{r'},\,(g^{r\cdot\sk})^{r'}\cdot (g^r)^{r'\sk'})=(g^{r\cdot r'},\,g^{r\cdot r'\cdot(\sk+\sk')}).$$
 \item[Filling the ciphertext] Having $\Enc{\sk}{1}=(g^r,\,g^{r\cdot\sk})$ one can compute  $$\Enc{\sk+\sk'}{C}=(g^r,\,g^{r\cdot\sk}\cdot C).$$
 \item[Partial decryption]
   Having $\Enc{\sk}{C}=(g^{r\cdot r'},\,g^{r\cdot r'\cdot(\sk+\sk')}C)$ and a private key $\sk '$ one can ``remove one layer of encryption'' and obtain 
    $$\Enc{\sk}{C}=\left(g^{r\cdot r'},\,\frac{g^{r\cdot r'\cdot(\sk+\sk')}C}{(g^{r\cdot r'})^{\sk'}}\right)=(g^{r\cdot r'},\,g^{r\cdot r'\cdot\sk}C).$$
 
\end{description}

For the sake of clarity we skip some technical details (i.e., choice of the group size, generators etc.) as well as full security discussion of this encryption scheme. 
Note that these are quite standard techniques used in many papers including \cite{URE,MURE}.

Similarly to previous papers (including \cite{Hubercik,PaniShi}) we utilize the following method: if we know that each user $v \in V$ has a value from an interval of  moderate size $\xi_v \in [0,\Delta]$  then the sum of values of all $\xi_v$'s cannot exceed  $ n\Delta$. Thus  one can find a discreet logarithm for $g^{\sum_{v\in V} \xi_v }$ even if finding 
a discreet logarithm of $g^r$ is not feasible if $r$ is a random element of $\mathbf{G}$. Using Pollard's Rho method this can be completed in average time $O(\sqrt{n\Delta})$.
 
\subsection{Protocol Description}\label{ssect:desc}
During the protocol, we assume that the aggregator $\Agg$ has a private key $\mathrm{sk}$, moreover each of the local aggregators $\lAgg_i$ has its own private key $\sk_i$. We also assume that there is a public parameter $g$, that is a generator of some finite group $\mathbf{G}$, in which Diffie-Hellman problem is hard. By $\Enc{\sk}{c}$ we denote the encryption structure introduced in Section~\ref{ssect:bild}.
Let us assume that each user $v$ has a private value $\xi_v$ from the range $[0,\,\Delta]$.
The final aim is to provide $\Agg$ the sum $\sum_{v\in V}\xi_v$ perturbed in such way that the privacy (expressed in terms of 
differential privacy) of all $v\in V$ is preserved.  Clearly, the privacy of users can be endangered both by reveling the output as well as by collecting 
information about the aggregation process.

 \begin{description}
\item[Setup] \hfill\\
\begin{itemize}
	\item $\Agg$ broadcasts to the local aggregators  $\Enc{\sk}{1}$.
	\item Each of the local aggregators $\lAgg_i$ constructs $\Enc{\sk+\sk_i}{1}$ and publishes it for all users from $V_i$.
	\end{itemize}
The setup phase is performed only once during network's lifetime. Moreover if needed, each $\lAgg_i$ may provide a non-interactive proof that
the operations were performed correctly and honestly~\cite{Greich,BlumNI}.\\

\item[Aggregation] \hfill\\

\begin{description}
\item[Algorithm for node $v$]\hfill\\
\begin{itemize}
	\item For each node $v'\in N(v)$ generate a random value $x^{v}_{v'}\in\mathbf{G}$.
	\item Using a private channel send each value $x^{v}_{v'}$ to the appropriate neighbor~$v'$.
	\item Having received all $x^{v'}_{v}$ from each of the neighbors, select random $r_v$ from $Geom^{\beta}(\alpha)$ and calculate 
                $$c_v=\sum_{v'\in N(v)}{x^{v'}_{v}}-\sum_{v'\in N(v)}{x^{v}_{v'}} +r_v+\xi_v.$$
        \item Compute $\mbox{Re}(\Enc{\sk+\sk_i}{g^{c_v}})$ and send it  to $\lAgg_i$.
\end{itemize}
An example of node's communication is shown in Fig.~\ref{fig:node}.
\item[Algorithm for local aggregator $\lAgg_i$]\hfill\\

\begin{itemize}
 \item Having received  $\Enc{\sk+\sk_i}{g^{c_v}}$ from all nodes from $V_i$,  compute
$$\Enc{\sk}{g^{c_v}}=\left(g^{r_i},\,\frac{g^{r_i(\sk+\sk_i)+c_v}}{g^{r_i\cdot \sk_i}}\right).$$ 
This operations result in obtaining shares  $$\Enc{\sk}{g^{c_{v_1}}}=(g^{r_{v_1}},g^{r_{v_1}\cdot \sk+c_{v_1}}),\,
\dots,\,\Enc{\sk}{g^{c_{v_l}}}=(g^{r_{v_l}},g^{{r_{v_l}}\cdot \sk+c_{v_l}})$$ of all  $l=|V_i|$ users from $|V_i|$. 
\item Compute
$$\Enc{\sk}{g^{c_{v_1}+\dots+c_{v_l}}}=\left(\prod_{i=1}^l g^{r_i},\,\prod_{i=1}^l{g^{r_{i}\sk+c_{v_i}}} \right)=\left(g^{\sum_{i=1}^l r_i},\,g^{(\sum_{i=1}^l r_i)\sk+\sum_{i=1}^l{c_{v_i}}} \right)~. $$
\item Send the value $\Enc{\sk}{g^{c_{v_1}+\dots+c_{v_l}}}$ to the aggregator $\Agg$.
\end{itemize}

\item[Final aggregation] \hfill\\
\begin{itemize}
 \item Having received the aggregated values from each $V_i$, for each of those values  $\Agg$ calculate $y_i= g^{\sum_{v\in V_i}{c_{v}}}$, using its private key $\sk$ for each 
$i=1,\ldots , k$.  Then compute
 
$$y = \prod_{i}^{k} y_i =  \prod_{i} g^{\sum_{v\in V_i}{c_{v_i}}}=g^{\sum_{v\in V}{c_{v_i}}}.$$

\item Then $\Agg$ compute discrete logarithm of $y$ as a final (perturbed) value being a sum of all $\sum_{v\in V}\xi_{v}$.
\end{itemize}

\end{description}
\end{description}
Note that the protocol depends on two security parameters $\beta$ and  $\alpha$. They strongly depend on the topology of the underlying graph.  
We discuss this issue in the next subsection.

\begin{figure}
\centering
\tikzstyle{peers}=[draw,circle,black,bottom color=white,
                  top color= white, text=black,minimum width=10pt]
\tikzstyle{lAgg}=[draw,circle,gray,bottom color=gray, top color=gray,
                       text=white,minimum width=20pt]

 \tikzstyle{myarrow}=[->,>=stealth,font=\scriptsize]
  \tikzstyle{brace}=[decorate,decoration={brace,amplitude=5pt}]
  \begin{tikzpicture}[xscale=0.7,yscale=0.7]
   \begin{scope}

	  \node (t1) [peers] at (2,1.5) {$v_1$};
    \node (t2) [peers] at (2.5,-2) {$v_2$};
    \node (t3) [peers]  at (-0.5,-3) {$v_3$};    
    \node (t4) [peers] at (-3.2,0) {$v_4$};
    \node (t5) [peers] at (-2.85,2.5) {$v_5$};
		\node (v)  [peers] at (0,0) {$v$};
		\node (A)  [lAgg] at (0.75,4) {$\lAgg$};
		\node[ellipse,draw=black,thick,dotted,fit=(A) (v) (t2) (t3) (t1),inner sep=0pt] {};

   \foreach \speer/\peer in {v/t1,v/t2,v/t3, v/t4, v/t5, v/A}
    \path (\speer) edge (\peer);
	\end{scope}

 \draw [myarrow] (v) -- (t1)  node[midway,above]{$x^v_{v_1}$};
 \draw [myarrow] (t1) -- (v)  node[midway,below]{$x_v^{v_1}$};
 \draw [myarrow] (v) -- (t2)  node[midway,above]{$x^v_{v_2}$};
 \draw [myarrow] (t2) -- (v)  node[midway,below]{$x_v^{v_2}$};
 \draw [myarrow] (v) -- (t3)  node[midway,above]{$x^v_{v_3}$};
 \draw [myarrow] (t3) -- (v)  node[midway,below]{$x_v^{v_3}$};
 \draw [myarrow] (v) -- (t4)  node[midway,above]{$x^v_{v_4}$};
 \draw [myarrow] (t4) -- (v)  node[midway,below]{$x_v^{v_4}$};
 \draw [myarrow] (v) -- (t5)  node[midway,above]{$x^v_{v_5}$};
 \draw [myarrow] (t5) -- (v)  node[midway,below]{$x_v^{v_5}$};

	 \draw [myarrow] (v) -- (A)  node[midway,above right]{$Re(\Enc{\sk+\sk_i}{g^{c_v}})$};
	
	  \end{tikzpicture}
  \caption{An example of communication in a single aggregation round from a perspective of node $v$. The dotted line marks the set of nodes assigned to a single local aggregator $\lAgg$. Note that neighbors may may have different local aggregators. }\label{fig:node}
\end{figure}
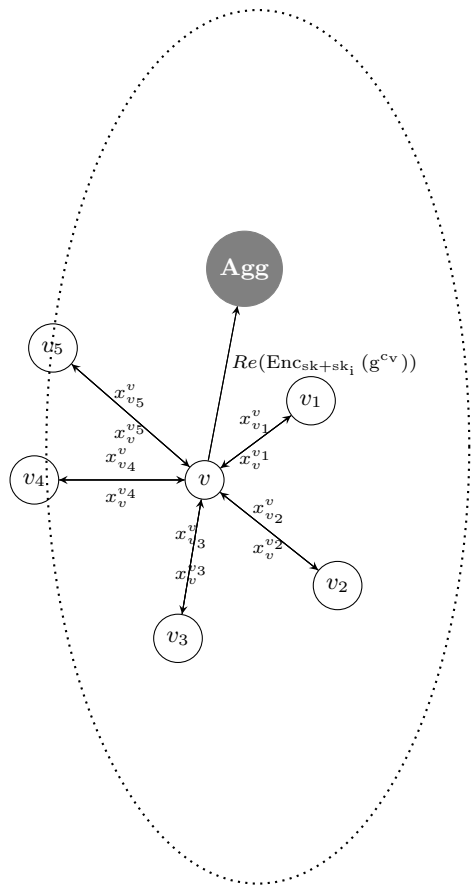

\subsection{Comparison and Analysis}\label{secdisc}
In this section we outline the  analysis of the presented aggregation protocol with respect to correctness, level of privacy provided and error of the result  obtained 
by the aggregator. The analysis is slightly more complicated since the parameters of the protocol strongly  depend on the underlying network. We argue however that 
they offer very good properties for wide classes of networks.

\paragraph{Correctness}

First, let us look at the result obtained by the aggregator $\Agg$ in the last step of the protocol. This is a discrete logarithm of $g^{\sum_{v\in V}{c_{v_i}}}$. 
Let us observe that 

\begin{align*}
\sum_{v\in V}{c_v}&=\sum_{v\in V}{\left(\sum_{v'\in N(v)}{x^{v'}_{v}}-\sum_{v'\in N(v)}{x^{v}_{v'}} +r_v+\xi_v\right)}\\
&=\sum_{v\in V}{\sum_{v'\in N(v)}{x^{v'}_{v}}}- \sum_{v\in V}{\sum_{v'\in N(v)}{x^{v}_{v'}}}+\sum_{v\in V}{\xi_v}+\sum_{v\in V}{r_v}=\sum_{v\in V}{\xi_v}+\sum_{v\in V}{r_v}.
\end{align*}

\noindent
The value  $\sum_{v\in V}{\xi_v}$ is the exact sum of values kept by nodes and sum of all the noises $\sum_{v\in V}{r_v}$. This leads to two conclusions. 
First, the result is correct. Second, retrieving the data using Pollard's Rho method (or even brute force method) is feasible since  
the absolute value of the first sum has to be smaller than $n\Delta$. One can easy see that the sum of added noises 
is of the magnitude $O(n)$ with high probability (as a sum of independent geometric distributions). 


\paragraph{Privacy protection}

We assume that the encryption scheme $\Enc{\sk}{}$ is \textit{semantically secure}. In particular after re-encryption  operation one cannot retrieve any non-trivial 
information about the plaintext without the private key $\sk$ possibly except some negligible probability $\eta\left(\lambda\right)$ with respect to the  key-length $\lambda$ or some other 
security parameters.  In particular, in our protocol, the local aggregator $\Agg_i$ cannot learn the contributions sent to $\Agg_j$ for $i\neq j$ without 
access to keys $\sk_j$ and $\sk$.

For the simplicity of analysis of the privacy protection let us consider the simplest case when $k=1$, i.e. there is only one aggregator.  
In such case we may assume $\Agg_1 = \Agg$. Let $V^{H} \subset V$ be the set of uncompromised users. Note that all neighboring users exchange a purely  
random values $x_{v}^{v'}$'s that finally cancel-out, however as long as they remain unknown to the adversary, they perfectly obfuscate the results sent to the aggregator 
(exactly in the same manner as the one-time pad cipher ).  This can be easily adopted to our protocol to get the following fact.

\begin{fact}\label{fuckt1}
 Let us assume that the adversary can control $\Agg$ and a subset of users \mbox{$V\setminus V^{H}$}. Let $\mathcal{S}$ be a connected component of the subgraph of 
 $\mathcal{G}=(V,E)$ induced by the subset $V^{H}$. Then, the adversary can learn nothing but $\sum_{v\in \mathcal{S}} (\xi_{v} + r_v)$ about the values $\xi_{v}$'s  
 from the execution of PAALEC for any $v\in V^{H}$.
\end{fact}

\begin{theorem}\label{qq}
Let us assume that PAALEC with parameter $\alpha = \exp(\frac{\epsilon}{\Delta})$ is executed in the network represented by a graph $\mathcal{G}=(V,E)$ 
and $\mathcal{G}'$ is a subgraph of $\mathcal{G}$ induced by the set of uncompromised users $V^{H}$. Moreover we assume that each user $v$ 
contributes a value $\xi_{v} \in [0,\Delta]$.  

If in each connected component $\mathcal{S}$ of $\mathcal{G}'$ there is a user $s$, such that its added noise  $r$ is taken from $Geom(\exp(\frac{\epsilon}{\Delta}))$, 
then PAALEC  preserves computational $(\varepsilon,\,0)$-differential privacy. 
\end{theorem}

\begin{proof}
 Let $\Xi = \sum_{s\in S} \xi_s$ and let $\Xi'$ be the same sum with changed a single value $\xi_s$. 
By the assumption about the range of the aggregated values we get $|\Xi'- \Xi|\leq \Delta$.  
 Let $r$ be a random variable taken from the symmetric geometric distribution  $Geom(\exp(\frac{\epsilon}{\Delta}))$. From Fact \ref{fuckt2}  we know that 
 $Pr[\Xi+r=k]$ may differ from  $Pr[\Xi'+r=k]$ by at most a multiplicative factor $\exp(\epsilon)$.  However, from Fact \ref{fuckt1} we know that the adversary may 
 learn nothing more than the sum of all values from the component $\mathcal{S}$. To complete the proof it is enough to recall that we assumed that probability of 
 gaining some other knowledge if  weak parameters of the cipher are chosen is at most negligible function  $\eta\left(\lambda\right)$.
\end{proof}
From this theorem follows next corollary.
\begin{corollary}\label{qq2}
If PAALEC is executed on a graph such that a subgraph  induced by the set of uncompromised users $V^{H}$ is connected
and with  probability at least $1-\delta$  at least one uncompromised users adds its value $r$ from $Geom(\exp(\frac{\epsilon}{\Delta}))$ 
then PAALEC  computationally preserves $(\varepsilon,\,\delta)$-differential privacy. 
 \end{corollary}

Translating into real terms  Theorem~\ref{qq} with Corollary~\ref{qq2} say if the connections between honest users are enough dense and we can somehow guarantee that
at least one honest node adds the noise,  the system is secure. The core of the problem is judge if a real-world networks are dense enough and what parameters 
of adding noise are sufficient. This problem is discussed in the next paragraph.

\paragraph{Accuracy}

The level of accuracy and security in this protocol strongly depends on the graph topology and chosen security parameters.
We will consider a random graph, where each of possible edge is independently added with probability $p$.
Moreover the adversary controls up to $n-m$ randomly chosen  users.

\begin{theorem}{\label{ssect:thmACC}}
Let us consider a random network with $n$ nodes. Each of possible  ${n \choose 2}$
connections (edges) is independently added to the network with probability
$p \geq \frac{8\log n}{n}$. Let $\mathcal{S}$  be a subgraph induced by a  subset
of at least  $m \geq n/2$ randomly chosen nodes. Then $\mathcal{S}$ is connected with
probability at least $1 - 1/n$.
\end{theorem}

Note that the  presented model boils down to the classic Erd\H{o}s-R\'{e}nyi model \cite{Luczak}. 
For the sake of completeness and to get explicit constants we present the proof in the appendix of the full version of our paper 
\cite{arXiv}.

From Theorem \ref{ssect:thmACC} we learn that a ``typical'' network of $n$ nodes with random connections 
such that the average number of neighbors is $8\log n = \Theta(\log n)$  is dense enough  even if the adversary is able to compromise
as much as $n/2$ nodes. 

If we have guaranteed at least $n/2$ honest (uncompromised and working) nodes one may note that the probability that none of them 
adds the noise is at least $(1-\beta)^{n/2}$. To have $(1-\beta)^{n/2} \leq \delta$ one needs to have $\beta$ such that 
$\log (1-\beta)  \leq \frac{2\log \delta}{n}$. Since  $\log (1+x) \leq x $ for $x>-1$ it is enough to use
$\beta \geq  \frac{2\log (1/\delta)}{n}$. Clearly the expected error cannot exceed  $2 \sqrt{\log (1/\delta)}$ for 
$\beta =  \frac{2\log (1/\delta)}{n}$. Using standard methods one can also show that the expected error is concentrated.

\paragraph{Remarks and  Extensions} 

We proved that the proposed protocol guarantees  a very good accuracy even facing a massive failures and compromising of nodes.
Half of nodes may failed or cooperate with the adversary (In fact this result can be generalized to any constant fraction of users).
The analysis and the model can be relaxed/extended in many directions. One can instantly observe that the analysis can be 
extended for smaller $\delta$ for the price of moderate increasing of the expected noise. 
Note that the value of  $\delta$ set to a celebrated magic constant $0.05$ seems to be definitely too big 
for practice. Indeed, this implies that one out of each $20$ may loos its privacy. 

We believe that this approach can be useful for other graphs-including those representing social networks. Note that if a graph 
guarantees a specific level of privacy then more dense graph (with some added edges) offers at least the same level of privacy. Thus 
it is enough if each users adds something like $\Theta(\log n)$ ``randomly'' chosen neighbors to protect the privacy in any network.

Note that our protocol is not immune against an adversarial nodes that sends incoherent random data. 
To the best of our knowledge all protocols of this type (including~\cite{Hubercik,PaniShi})  are prone to so called \textit{contaminating} attacks. 
To mitigate this problem as in other cases one may apply orthogonal methods presented in~\cite{sia}.

 \section{Previous and Related Work}\label{sect:prev} 
Data aggregation in distributed networks has been thoroughly studied due to practical importance of such protocols. Measuring the target environment,  aggregating data and rising alarm are arguably three most important functionalities of distributed sensing networks, and with the increased number of personal mobile devices, the aggregation becomes of greatest interest among the three. Exemplary protocols that do not address security nor privacy may be found in~\cite{unsecureAgg,tag}, with the latter being often presented as a model aggregation algorithm. \\
There are several settings considering data aggregation. They differ in both, the abilities and constraints of the nodes performing the aggregation, as well as the issues that the algorithm addresses. Some of the adversities that may be addressed include data confidentiality (i.e., protecting the data from disclosure), privacy of the nodes (inability to learn exact values of each node), node failure and spontaneous node joining the network as well as data poisoning (i.e., injecting malicious data by the adversary that allows them to significantly influence the outcome of the algorithm or learning more information about the execution that they would not gain when following the protocol honestly). \\
Our paper follows the model considered in~\cite{Hubercik}, where the nodes have constrained abilities and their energy pool is limited. Authors present a privacy preserving aggregation protocol that assumes malicious aggregator, moreover they claim tolerance for failures and joins, hence addressing majority of the issues. Similar problems that focus on narrower range of properties have been also studied in~\cite{PaniShi,Rastogi}. In~\cite{PDA,PDA2} authors present some aggregation protocols that preserve privacy, however they do not consider dynamic changes inside of the network. The latter also considers data poisoning attacks, however the authors do not provide rigid proofs. A different approach was presented in~\cite{jajodia,6171193}, where the authors present a framework for some aggregation functions and consider the confidentiality of the result, however leaving nodes' privacy out of scope of their papers. On the other hand, there is bulk of research that focuses on fault tolerance that leaves 
privacy and security issues either out of scope or just mentioned, not keeping it as a priority. Examples of such work may be found in~\cite{Feng2011451,Jhumka20141789,4147120}.
In~\cite{HubercikBound} the authors present an asymptotic lower bound on the error of the aggregation that preserves privacy, showing that in order to reduce the errors, one has to resign from perfect privacy and focus rather on computational variant of the privacy preservation.\\
An example of work on secure data aggregation in stronger models may be found in~\cite{jawurek,rottondi2013}, where the authors consider data aggregation in a smart grid. 
Another fruitful branch of the research on data aggregation considers data aggregation in vehicular ad~hoc networks (VANET). The research in this field is motivated by the increasing number of ,,smart-cars'' with internal computational unit. One of the first works addressing this issue was~\cite{SOTIS,Nadeem,sot2}. A practical scenario for data aggregation in VANET has been presented in~\cite{parking}. The security issue in VANET data-aggregation has been mentioned in~\cite{riv,sas}. A survey of the known protocols has been performed in~\cite{Mohanty2012922}.
One may note that retrieving  encrypted or blinded data by one entity, that requires cooperation of others is similar to cryptographic secret-sharing. Some of the most important work on secret sharing may be found in~\cite{Benaloh,sssurvey}, however in our paper we draw from the Universal Re-encryption method presented in~\cite{URE}.

\section{Conclusions}\label{sect:conclusion}
In our paper we provided a precise analysis of accuracy of the data aggregation protocol presented in~\cite{Hubercik}. We have shown that in many cases its accuracy may not be sufficient even if the number of faults is moderate. 
We constructed another fault tolerant, privacy preserving aggregation protocol that offers much better precision. In order to obtain this, we allowed a moderate communication between the nodes. This assumption deviates from the classic model. 

We believe that our approach and security model is justified in many real-life scenarios, however much research is left to be done in the field. First of all, our protocol as well as all other similar protocols we are aware of, is not immune against so called data poisoning attack. Another problem is finding solution for statistics other than sum. Authors of aggregating schemes usually limit the scope of their work to sum, product and average of the values of all nodes in the network. In many cases we need however other statistics, e.g. minimum or the median.  We suppose that finding more general statistics with guaranteed privacy of individuals is possible using methods explored in e-voting protocols. They however are very demanding in terms of required resources. 
  From the theoretical point of view the important question is about the possible trade-offs between privacy protection, volume of communication and possible accuracy of the results of aggregation.

%
%
%
\bibliographystyle{splncs-url}
\bibliography{bibliography} 
 
\section*{Appendix}\label{sect:App}
\renewcommand{\thesubsection}{\Alph{subsection}}

\subsection{Proof of Theorem~\ref*{ssect:thmEY}.}

\begin{proof}
Consider Binary Protocol described in~\ref{BP}. We aim to give a precise formula for the expected value of the number of noises added in this protocol. For simplicity we assume that $n$ is the power of $2$. We also assume that $\kappa$ leaves have failed, and they are uniformly chosen from all $n$ leaves. We will use random variables $X_i$ to denote the number of segments (on $i$-th level of the tree) corresponding to subset of users with no failures. We will also use random variable $X^*_i$ to denote the number of aggregating nodes on the $i$-th level of the tree. Let us begin with stating and proving the following
\begin{lemma}{\label{ssect:lemmaEXi*}}
Consider Binary Protocol with fixed $\kappa$ and $n$. We call a node an \textit{aggregating node}, if it is used by the aggregator to obtain a sum of data from some subset of users. We have the following formula for $i \geqslant 1$
$$
EX_i^* = EX_i - 2EX_{i-1} = 2^i\cdot \left(\frac{\binom{n-\frac{n}{2^{i}}}{\kappa}}{\binom{n}{\kappa}} - \frac{\binom{n-\frac{n}{2^{i-1}}}{\kappa}}{\binom{n}{\kappa}} \right).
$$
\end{lemma}
\begin{proof}
First of all, we will call a segment in the Binary Protocol tree \textit{clean} if and only if there are no fails in this segment. Each node in the tree corresponds to a specific segment, according to Binary Protocol rules. See that on a certain tree level, all nodes correspond to segments of the same size, noted here by $|B_i|$. Throughout this reasoning we will call the 'root level' 0, children of the root are on level $1$ and so on, up to level $\log_2(n)$ which is the 'leaves level'.

\par

Data of each user is aggregated in exactly one node, which belongs to some $i$th level and corresponds to a specific segment. This user generates a geometric noise with probability $\beta_i$, where:
$$
\beta_i = \min\left(\frac{1}{|B_i|}\ln\left(\frac{\log_2(n)+1}{\delta}\right), 1\right).
$$
We want to know an expected value of the number of noises generated throughout the whole protocol.

\par
To do this, first we denote the number of 'clean' segments of size $|B_i|$ (corresponding to nodes on $i$th level of the tree) by a random variable $X_i$. See that $X_i \in \{0, 1, \ldots, 2^i\}$. Furthermore, we see that:
$$
X_i = \sum_{j=1}^{2^i}X_{i,j},
$$
where 
$$
X_{i,j} = \begin{cases}
1, \qquad  \text{if segment j on level i has no fails},\\
0, \qquad \text{otherwise}.
\end{cases}
$$
This, and the fact that $EX_{i,j} = EX_{i,k}$ for every $j, k \in {0, \ldots, 2^i}$, allows us to use linearity of expectation to calculate $EX_i$:
\begin{equation}\label{ssect:EXi}
EX_i = E\sum_{j=1}^{2^i}X_{i,j} = \sum_{j=1}^{2^i}EX_{i,j} = 2^i EX_{i,1} = 2^i \cdot P(X_{i,1} = 1).
\end{equation}
Now see that
$$
P(X_{i,1} = 1) = \frac{\binom{n-|B_i|}{\kappa}}{\binom{n}{\kappa}},
$$
and also $|B_i| = \frac{n}{2^i}$, thus plugging these to (\ref{ssect:EXi}) we get
\begin{equation}\label{ssect:EXiFormula}
EX_i = 2^i \cdot \frac{\binom{n-\frac{n}{2^i}}{\kappa}}{\binom{n}{\kappa}}.
\end{equation}

Now let us consider the number of segments which really aggregate the data. See that if a node is an aggregating one, that means that it corresponds to a clean segment, but its parent does not correspond to a clean segment. We denote the number of aggregating nodes on $i$th level by $X_i^*$, we can also see that $X_i^* = X_i - 2X_{i-1}$, where $i \in \{1, \ldots, \log_2(n)\}$. \\
There are $X_i$ clean nodes on $i$th level but we have to subtract all the clean nodes from higher level of the tree multiplicated by $2$, because each of these clean nodes on a higher level is parent to two nodes on $i$th level, which are therefore not an aggregating nodes, because their parent is clean. That gives us
$$
EX_i^* = EX_i - 2EX_{i-1} = 2^i\cdot \left(\frac{\binom{n-\frac{n}{2^{i}}}{\kappa}}{\binom{n}{\kappa}} - \frac{\binom{n-\frac{n}{2^{i-1}}}{\kappa}}{\binom{n}{\kappa}} \right),
$$
which completes the proof of this lemma. \qed
\end{proof}

Lemma \ref{ssect:lemmaEXi*} gives us an explicit formula for $EX_i^*$. Now, when we have a formula for the expected value of the number of aggregating nodes on each level, we can proceed to calculating the expected value of the number of geometric noises generated during the Binary Protocol.
\par
Let $Y_i$ be a random variable which denotes the number of noises generated on $i$th level of the tree. On $i$th level we aggregate $X_i^*$ segments, each of these segments have $2^{\log_2(n)-i}$ users and each of these users generates geometric noise with probability $\beta_i$. Therefore we have  $Y_i \sim Bin\left(2^{\log_2(n)-i} \cdot X_i^*, \beta_i\right)$, where $Bin(n,p)$ denotes binomial distribution. After observing this, we can see that
$$
EY_i = EX_i^* \cdot 2^{\log_2(n)-i} \cdot \beta_i. 
$$
Every user is aggregated only on one level, so if we take a sum over all levels of the tree, we will get all the noises generated during the Binary Protocol. Let $Y$ be a random variable that denotes the number of noises generated. We have
$$
Y = \sum_{i=0}^{\log_2(n)}Y_i,
$$
and we can also safely assume that if $\kappa > 0$, then $Y_0 = 0$, because if at least one user has failed, then we cannot possibly aggregate all users in the root of the tree. Furthermore, using linearity of expectation and well known expected value for Binomial distribution we have
$$
EY = \sum_{i=1}^{\log_2(n)} EX_i^* \cdot 2^{\log_2(n)-i} \cdot \beta_i = \sum_{i=1}^{\log_2(n)} (EX_i - 2EX_{i-1}) \cdot 2^{\log_2(n)-i} \cdot \beta_i.
$$
After simple algebraic manipulations we can get
\begin{align*}
EY &= \sum_{i=1}^{\log_2(n)} EX_i \cdot 2^{\log_2(n)-i} \cdot \beta_i - \sum_{i=1}^{\log_2(n)} 2EX_{i-1} \cdot 2^{\log_2(n)-i} \cdot \beta_i = \\
&= \sum_{i=1}^{\log_2(n)} EX_i \cdot 2^{\log_2(n)-i} \cdot \beta_i - \sum_{i=0}^{\log_2(n)-1} EX_{i} \cdot 2^{\log_2(n)-i} \cdot \beta_{i+1} = \\
&= EX_{\log_2(n)} \cdot \beta_{\log_2(n)} - n \beta_1 EX_0 + \sum_{i=1}^{\log_2(n)-1} EX_i \cdot 2^{\log_2(n)-i} \cdot \left(\beta_i-\beta_{i+1}\right).
\end{align*}
Also, as $\kappa > 0$, we have $X_0 = 0$ with probability $1$. These facts yield the following result
\begin{align*}
EY &= EX_{\log_2(n)} + \sum_{i=1}^{\log_2(n)-1} EX_i \cdot 2^{\log_2(n)-i} \cdot \left(\beta_i - \beta_{i+1}\right) = \\
&= n \cdot \frac{\binom{n-1}{\kappa}}{\binom{n}{\kappa}} + \sum_{i=1}^{\log_2(n)-1} 2^i \cdot \frac{\binom{n-\frac{n}{2^i}}{\kappa}}{\binom{n}{\kappa}} \cdot 2^{\log_2(n)-i} \cdot  \left(\beta_i - \beta_{i+1}\right) = \\
&= n-\kappa + n\cdot \sum_{i=1}^{\log_2(n)-1}\left(\frac{\binom{n-\frac{n}{2^i}}{\kappa}}{\binom{n}{\kappa}} \cdot  \left(\beta_i - \beta_{i+1}\right)\right).
\end{align*}
This gives us a formula for calculating $EY$ and completes the proof of this theorem. \qed
\end{proof}

\subsection{Proof of Lemma \ref*{ssect:lemmaEY}.}
\begin{proof}
We fix $\delta = 0.05$. First observe that for $2^4 \leqslant n \leqslant 2^{21}$ we have 
\begin{equation*}
\beta_{\log_2(n)}=\beta_{\log_2(n)-1}=\beta_{\log_2(n)-2} = 1,
\end{equation*}
as for these levels we have $\frac{1}{|B_i|}\cdot \ln(\log(n)+1) > 1$. This means that users aggregated in segments of length $1$ and $2$ generate noise with probability $1$. Furthermore, for $i \leqslant (\log_2(n)-3)$ we have $\beta_{i} < 1$.  Also, for $i \leqslant (\log_2(n)-4)$ we have
\begin{equation*}
\frac{\beta_{i+1}}{\beta_i} = \frac{|B_i|}{|B_{i+1}|} = 2.
\end{equation*}
Another observation is that we can get an upper bound for $\frac{\binom{n-\frac{n}{2^i}}{\kappa}}{\binom{n}{\kappa}}$ in a following way
\begin{align*}
\frac{\binom{n-\frac{n}{2^i}}{\kappa}}{\binom{n}{\kappa}} &= \frac{(n-\frac{n}{2^i})! \cdot (n-\kappa)!}{(n-\frac{n}{2^i}-\kappa)! \cdot n!} = \\
&= \frac{(n\cdot \frac{2^i-1}{2^i})\cdot (n\cdot \frac{2^i-1}{2^i}-1) \cdot \ldots \cdot (n\cdot \frac{2^i-1}{2^i}-\kappa +1)}{n\cdot (n-1) \cdot \ldots \cdot (n-\kappa+1)} = \\
&= \left(\frac{2^i-1}{2^i}\right)^{\kappa} \cdot \frac{n\cdot (n\cdot-\frac{2^i}{2^i-1}) \cdot \ldots \cdot (n-(\kappa -1) \cdot \frac{2^i}{2^i-1})}{n\cdot (n-1) \cdot \ldots \cdot (n-\kappa+1)} \leqslant \\ 
&\leqslant \left(\frac{2^i-1}{2^i}\right)^{\kappa} = \left(1-\frac{1}{2^i}\right)^{\kappa} = \left(\left(1-\frac{1}{2^i}\right)^{2^i}\right)^{\frac{\kappa}{2^i}} \leqslant e^{-\frac{\kappa}{2^i}},  
\end{align*}
where the last inequality comes from the fact that $(1-x) \leqslant e^{-x}$. We can use all these observations to obtain a lower bound. Let $\beta^* = \ln\left(\frac{\log_2(n)+1}{\delta}\right)$. Then we have
\begin{align*}
EY &= n-\kappa + n\cdot \sum_{i=1}^{\log_2(n)-1}\left(\frac{\binom{n-\frac{n}{2^i}}{\kappa}}{\binom{n}{\kappa}} \cdot  \left(\beta_i - \beta_{i+1}\right)\right) = \\
&= n-\kappa -   n\cdot \left(\sum_{i=1}^{\log_2(n)-4}\left(\frac{\binom{n-\frac{n}{2^i}}{\kappa}}{\binom{n}{\kappa}} \cdot  \beta_i\right) + \frac{\binom{n-8}{\kappa}}{\binom{n}{\kappa}} \cdot  \left(1-\beta_{\log_2(n)-3} \right) \right) \geqslant \\ 
&\geqslant n-\kappa -   n\cdot \left(\sum_{i=1}^{\log_2(n)-4}\left(e^{-\frac{\kappa}{2^i}} \cdot  \beta_i\right) + e^{\frac{8\kappa}{n}} \cdot  \left(1-\beta_{\log_2(n)-3} \right) \right) \geqslant \\
&\geqslant n-\kappa -   n\cdot \left(\sum_{i=1}^{\log_2(n)-4}\left(e^{-\frac{\kappa}{2^{\log_2(n)-4}}} \cdot  \beta_i\right) + e^{\frac{8\kappa}{n}} \cdot  \left(1-\beta_{\log_2(n)-3} \right) \right) = \\
&= n-\kappa -   n\cdot \left( e^{-\frac{16 \kappa}{n}} \cdot \frac{\beta^*}{n} \cdot \sum_{i=1}^{\log_2(n)-4}\left(2^i\right) + e^{\frac{8\kappa}{n}} \cdot  \left(1-\beta_{\log_2(n)-3} \right) \right) = \\
&= n-\kappa -   n\cdot \left( e^{-\frac{16 \kappa}{n}} \cdot \frac{\beta^*}{n} \cdot \left(\frac{n}{8} - 2\right) + e^{\frac{8\kappa}{n}} \cdot  \left(1-\frac{\beta^*}{8} \right) \right) \geqslant \\
&\geqslant n-\kappa -   n\cdot \left( e^{-\frac{16\kappa}{n}} \cdot \frac{\beta^*}{8} + e^{\frac{8\kappa}{n}} \cdot  \left(1-\frac{\beta^*}{8} \right) \right) = \\
&= n-\kappa - n\cdot \left( e^{-\frac{8\kappa}{n}} + \frac{\beta^*}{8} \cdot \left( e^{-\frac{16\kappa}{n}} - e^{-\frac{8\kappa}{n}} \right)\right).
\end{align*}
Which gives our lower bound for $EY$ and finishes the proof of this lemma. \qed
\end{proof}

\subsection{Proof of Theorem \ref*{ssect:thmEZ}.}
\begin{proof}
We are interested in the absolute sum of $m$ noises, to estimate the error in Binary Protocol. First, let $Z$ be a random variable that denote the sum of noises. See that
$$
Z = \sum_{i=1}^m Z_i,
$$
where $Z_i$ is a random variable with distribution Geom($\alpha$), where $\alpha = e^{\frac{\epsilon}{\log_2(n)+1}}$.

Let $\varphi_{Z_i}(t)$ denotes the characteristic function of $Z_i$. We have
$$
\varphi_{Z_i}(t) = \frac{(\alpha-1)^2}{\alpha^2 - \alpha(e^t+e^{-t}) + 1} = \frac{(\alpha-1)^2}{\alpha^2 - 2\alpha \cos{t} + 1}.
$$
Let $\varphi_{Z}(t)$ denote the characteristic function of $Z$. As $Z_i$ are i.i.d. random variables, we get
$$
\varphi_{Z}(t) = \left(\varphi_{Z_1}\right)^m = \left(\frac{(\alpha-1)^2}{\alpha^2 - 2\alpha \cos{t} + 1}\right)^m.
$$
We will use techniques comprehensively described in~\cite{pinelis} to calculate expected value of $|Z|$. We have a following
\begin{fact} (From \cite{pinelis})
$$
\varphi_{Z_+}(t) = Ee^{itZ_+} = \frac{1}{2}[1+\varphi_{Z}(t)] + \frac{1}{2\pi i} \int\displaylimits_{-\infty}^{\infty}\left[\varphi_{Z}(t+u)-\varphi_{Z}(u)\right]\frac{du}{u},
$$
\end{fact}
where $Z_+$ denotes $\max(0,Z)$, and the integral is understood in the principal value sense (see~\cite{pinelis}). Now see that
$$
|Z| = Z_+ + Z_- = Z_+ + (-Z_+) = 2Z_+,
$$
which is true for symmetric $Z$. Fortunately, this is the case here. Furthermore, we have
\begin{equation}\label{ssect:E|Z|}
E|Z| = 2EZ_+ = 2\frac{\varphi_{Z_+}'(0)}{i}.
\end{equation}
We have to calculate the derivative of $\varphi_{Z_+}(t)$ at $0$. It can be done in the following way
\begin{align}\label{ssect:phiat0}
\begin{split}
\varphi_{Z_+}'(0) &=  \frac{\varphi_{Z}'(0)}{2} + \frac{d}{dt}\left(\frac{1}{2\pi i} \int\displaylimits_{-\infty}^{\infty}\left[\varphi_{Z}(t+u)-\varphi_{Z}(u)\right]\frac{du}{u}\right)\left(0\right) \\
&= \frac{1}{2\pi i} \left(\int\displaylimits_{-\infty}^{\infty}\left[\varphi_{Z}'(t+u)\right]\frac{du}{u}\right)(0) = \frac{1}{2\pi i} \int\displaylimits_{-\infty}^{\infty}\left[\varphi_{Z}'(u)\right]\frac{du}{u}.
\end{split}
\end{align}
We used the fact that $\varphi_{Z}'(0) = 0$, because $Z$ is symmetric. Moreover, because $EZ$ exists, then $E|Z|$ also has to exist. That is why the integral has to be finite, so we were able to use Lebesgue theorem to swap order of derivation and integration. We can derive $\varphi_{Z}(t)$ which yields the following
\begin{equation}\label{ssect:f'(t)}
\varphi_{Z}'(t) = \frac{-2\cdot \alpha \cdot m \cdot \sin{t} \cdot \left(\alpha-1\right)^{2m}}{\left(\alpha^2 - 2\alpha \cos{t} + 1\right)^{m+1}}.
\end{equation}
Combining (\ref{ssect:E|Z|}), (\ref{ssect:phiat0}), (\ref{ssect:f'(t)}) and observing that $\varphi_{Z}'(t)$ is an even function, we obtain the following formula for $E|Z|$
$$
E|Z| = \int\displaylimits_{0}^{\infty} \frac{4\cdot \alpha \cdot m \cdot \sin{t} \cdot \left(\alpha-1\right)^{2m}}{t \cdot \pi \cdot \left(\alpha^2 - 2\alpha \cos{t} + 1\right)^{m+1}}dt,
$$
which completes the proof of this theorem. \qed
\end{proof}

\subsection{Proof of Lemma \ref*{ssect:lemmaEZ}.}
\begin{proof}
Let us define $\omega(t)$
$$
\omega(t) = \frac{4\cdot \alpha \cdot m \cdot \sin{t} \cdot \left(\alpha-1\right)^{2m}}{\pi \cdot \left(\alpha^2 - 2\alpha \cos{t} + 1\right)^{m+1}}
$$
We have 
$$
E|Z| = \int\displaylimits_{0}^{\infty} \frac{4\cdot \alpha \cdot m \cdot \sin{t} \cdot \left(\alpha-1\right)^{2m}}{t \cdot \pi \cdot \left(\alpha^2 - 2\alpha \cos{t} + 1\right)^{m+1}}dt = \int\displaylimits_{0}^{\infty} \frac{\omega(t)}{t}.
$$
One can easily see that $\omega(t)$ is periodic with period $2\pi$. We can therefore consider splitting the integral into $[2k \pi, 2(k+1)\pi]$ intervals and try to find a good lower bound for this integral. We have
$$
E|Z| = \sum_{k=0}^{\infty} \left( \int\displaylimits_{2k\pi}^{2(k+1)\pi} \frac{4\cdot \alpha \cdot m \cdot \sin{t} \cdot \left(\alpha-1\right)^{2m}}{t \cdot \pi \cdot \left(\alpha^2 - 2\alpha \cos{t} + 1\right)^{m+1}}dt \right).
$$
Consider any of these integrals for $k > 0$
\begin{equation}\label{ssect:lowerbound1}
\int\displaylimits_{2k\pi}^{2(k+1)\pi} \frac{4\cdot \alpha \cdot m \cdot \sin{t} \cdot \left(\alpha-1\right)^{2m}}{t \cdot \pi \cdot \left(\alpha^2 - 2\alpha \cos{t} + 1\right)^{m+1}}dt \geqslant 0.
\end{equation}
We will now explain why this inequality holds. First, observe that function $\omega(t)$ is an odd function on interval $[2k\pi,2(k+1)\pi]$. One can easily see, that $\omega(t)$ is positive on $[2k\pi,2k\pi+\pi]$ and negative on $[2k\pi + \pi, 2(k+1)\pi]$. Furthermore, the absolute value of $\frac{\omega(t)}{t}$ is greater on the first half of the interval, because of the decreasing factor $\frac{1}{t}$. This yields (\ref{ssect:lowerbound1}), which is true for all these intervals, and we will use it for all $k > 0$, so that leaves us with
\begin{align*}
E|Z| &= \sum_{k=0}^{\infty} \left( \int\displaylimits_{2k\pi}^{2(k+1)\pi} \frac{4\cdot \alpha \cdot m \cdot \sin{t} \cdot \left(\alpha-1\right)^{2m}}{t \cdot \pi \cdot \left(\alpha^2 - 2\alpha \cos{t} + 1\right)^{m+1}}dt \right) \\ &\geqslant \int\displaylimits_{0}^{2\pi} \frac{4\cdot \alpha \cdot m \cdot \sin{t} \cdot \left(\alpha-1\right)^{2m}}{t \cdot \pi \cdot \left(\alpha^2 - 2\alpha \cos{t} + 1\right)^{m+1}}dt.
\end{align*}
Plotting this function shows that almost all of the mass is concentrated around $0$, especially for $\alpha$ close to $1$. We could use the lower bound (\ref{ssect:lowerbound1}), however there is no point using it on the whole interval, because we would obtain trivial inequality $E|Z| \geqslant 0$. It requires slightly more subtle handling. Clearly, we could use (\ref{ssect:lowerbound1}) for any interval of type $[\pi - x, \pi + x]$, for $x \leqslant \pi$. This yields the following
\begin{align*}
E|Z| &\geqslant \int\displaylimits_{0}^{2\pi} \frac{4\cdot \alpha \cdot m \cdot \sin{t} \cdot \left(\alpha-1\right)^{2m}}{t \cdot \pi \cdot \left(\alpha^2 - 2\alpha \cos{t} + 1\right)^{m+1}}dt  \\
&\geqslant \int\displaylimits_{0}^{\eta_{\alpha,m}} \frac{4\cdot \alpha \cdot m \cdot \sin{t} \cdot \left(\alpha-1\right)^{2m}}{t \pi\left(\alpha^2 - 2\alpha \cos{t} + 1\right)^{m+1}}dt + \int\displaylimits_{2\pi-\eta_{\alpha,m}}^{2\pi} \frac{4\cdot \alpha \cdot m \cdot \sin{t} \cdot \left(\alpha-1\right)^{2m}}{t \pi \left(\alpha^2 - 2\alpha \cos{t} + 1\right)^{m+1}}dt,
\end{align*}
which is true for every $\eta_{\alpha,m} \in [0,\pi]$. Now see that if $\eta_{\alpha,m} < \frac{\pi}{2}$, we can bound the first integral in a following way
\begin{equation}\label{ssect:intA}
\int\displaylimits_{0}^{\eta_{\alpha,m}} \frac{4\alpha m \cdot \sin{t} \cdot \left(\alpha-1\right)^{2m}}{t \cdot \pi \cdot \left(\alpha^2 - 2\alpha \cos{t} + 1\right)^{m+1}}dt \geqslant \int\displaylimits_{0}^{\eta_{\alpha,m}} \frac{4 \alpha m \cdot \cos{t} \cdot \left(\alpha-1\right)^{2m}}{\pi \cdot \left(\alpha^2 - 2\alpha \cos{t} + 1\right)^{m+1}}dt,
\end{equation}
which follows from the fact that $x \leqslant \tan{x}$ for $x \in [0,\frac{\pi}{2})$. Furthermore
\begin{equation}\label{ssect:intB}
\int\displaylimits_{2\pi-\eta_{\alpha,m}}^{2\pi} \frac{4 \alpha m \cdot \sin{t} \cdot \left(\alpha-1\right)^{2m}}{t \cdot \pi \cdot \left(\alpha^2 - 2\alpha \cos{t} + 1\right)^{m+1}}dt \geqslant \int\displaylimits_{2\pi-\eta_{\alpha,m}}^{2\pi} \frac{4\alpha m \sin{t} \left(\alpha-1\right)^{2m}}{t \cdot \pi \cdot \left(\alpha - 1\right)^{2m+2}}dt,
\end{equation}
which comes from plugging $1$ instead of $\cos{t}$, which makes the function greater in terms of absolute value, but as it is negative on this interval, it yields a lower bound.
The function from (\ref{ssect:intA}) has an explicit anti-derivative. On the other hand, in (\ref{ssect:intB}) we have, in fact, an integral of $\frac{\sin{t}}{t}$ multiplied by a constant depending on $\alpha$ and $m$. There also still remains a problem of choosing $\eta_{\alpha, m}$. First we can observe that, for small enough $\eta_{\alpha,m}$ we have
$$
\int\displaylimits_{2\pi-\eta_{\alpha,m}}^{2\pi} \frac{\sin{t}}{t}dt \geqslant -\frac{\eta_{\alpha,m}^2}{10}.
$$
Obviously this holds for $\eta_{\alpha,m} = 0$. Let $Si(x)$ denote the antiderivative of $\frac{\sin{x}}{x}$. After derivating left side we obtain
\begin{align*}
\frac{d\left(Si(2\pi)-Si(2\pi-\eta_{\alpha,m})\right)}{d\eta_{\alpha,m}} &= -\frac{d\left(Si(2\pi-\eta_{\alpha,m}\right)}{d\eta_{\alpha,m}} = \frac{\sin\left(2\pi - \eta_{\alpha,m}\right)}{2\pi-\eta_{\alpha,m}} = \\
&= -\frac{\sin(\eta_{\alpha,m})}{2\pi - \eta_{\alpha,m}} \geqslant -\frac{\eta_{\alpha,m}}{2\pi - \eta_{\alpha,m}}.
\end{align*}
Derivating the right side yields $-0.2\eta_{\alpha,m}$. We can check when the left side is greater than the right side
$$
-\frac{\eta_{\alpha,m}}{2\pi - \eta_{\alpha,m}} \geqslant -0.2\eta_{\alpha,m} \iff \eta_{\alpha,m} \leqslant 2\pi-5
$$
So for $\eta_{\alpha,m} \leqslant \left(2\pi-5\right)$ we have
$$
\int\displaylimits_{2\pi-\eta_{\alpha,m}}^{2\pi} \frac{\sin{t}}{t}dt \geqslant -\frac{\eta_{\alpha,m}^2}{10}
$$
Now we pick $\eta_{\alpha,m}$ so that 
$$
-0.1\eta_{\alpha,m}^2 \cdot \frac{4\alpha m}{\pi(\alpha-1)^2} = -0.1.
$$
That gives us 
$$
\eta_{\alpha,m} = \sqrt{\frac{\pi(\alpha-1)^2}{4\alpha m}}.
$$
Plugging it all to our formula for expected magnitude of noise yields 
$$
E|Z| \geqslant \int\displaylimits_{0}^{\eta_{\alpha,m}} \frac{4\cdot a \cdot m \cdot \cos{t} \cdot \left(\alpha-1\right)^{2m}}{\pi \cdot \left(\alpha^2 - 2\alpha \cos{t} + 1\right)^{m+1}}dt - 0.1.
$$
We are now interested in the lower bound for this integral. One can see that
$$
\int\displaylimits_{0}^{\eta_{\alpha,m}} \frac{4\cdot a \cdot m \cdot \cos{t} \cdot \left(\alpha-1\right)^{2m}}{\pi \cdot \left(\alpha^2 - 2\alpha \cos{t} + 1\right)^{m+1}}dt \geqslant \int\displaylimits_{0}^{\eta_{\alpha,m}} \frac{4\cdot a \cdot m \cdot \cos({\eta_{\alpha,m}}) \cdot \left(\alpha-1\right)^{2m}}{\pi \cdot \left(\alpha^2 - 2\alpha \cos{t} + 1\right)^{m+1}}dt.
$$
This inequality is just plugging the smallest possible value of cosine on this interval. Furthermore, we have
$$
\int\displaylimits_{0}^{\eta_{\alpha,m}} \frac{4\alpha m \cdot \cos({\eta_{\alpha,m}}) \cdot \left(\alpha-1\right)^{2m}}{\pi\left(\alpha^2 - 2\alpha \cos{t} + 1\right)^{m+1}}dt \geqslant \int\displaylimits_{0}^{\eta_{\alpha,m}} \frac{4\alpha m \cdot \left(1-\frac{\eta_{\alpha,m}^2}{2}\right) \cdot \left(\alpha-1\right)^{2m}}{\pi\left(\alpha^2 - 2\alpha\cdot\left(1-\frac{t^2}{2}\right) + 1\right)^{m+1}}dt.
$$
This bound comes from the fact that $\cos{t} \geqslant \left(1-\frac{t^2}{2}\right)$. Let us call the integrand function $g(t)$. This function has a following anti-derivative $G(t)$:
$$
G(t) = \frac{4(\alpha-1)^{2m-2} \alpha m  t \left(1+\frac{\alpha  t^2}{(\alpha-1)^2}\right)^m    \left(1-\frac{\eta_{\alpha,m}^2}{2}\right) {}_2F_1\left(\frac{1}{2},1+m,\frac{3}{2},-\frac{\alpha\cdot t^2}{(\alpha-1)^2}\right) }{\left(\alpha^2 +\alpha  (t^2 - 2) + 1\right)^{m}\cdot \pi},
$$
where the ${}_2F_1(a,b,c,z)$ denotes ordinary hypergeometric function (see~\cite{Hypergeometric2F1}). One can easily see, that $G(0) = 0$. That leaves us with
$$
E|Z| \geqslant G(\eta_{\alpha,m}) - 0.1.
$$
Function $G(\eta_{\alpha,m})$ is quite complicated, but we can greatly simplify it. Let us begin with taking some of the $G(\eta_{\alpha,m})$ factors
\begin{align*}
\frac{(\alpha-1)^{2m-2}\cdot \left(1+\frac{\alpha \cdot \eta_{\alpha,m}^2}{(\alpha-1)^2}\right)^m}{\left(\alpha^2 +\alpha \cdot (\eta_{\alpha,m}^2 - 2) + 1\right)^{m}} &= \frac{(\alpha-1)^{-2}\cdot \left(1+\frac{\alpha \cdot \eta_{\alpha,m}^2}{(\alpha-1)^2}\right)^m}{\left(\frac{\alpha^2}{(\alpha-1)^{2}} +\frac{\alpha}{(\alpha-1)^{2}} \cdot (\eta_{\alpha,m}^2 - 2) + \frac{1}{(\alpha-1)^{2}}\right)^{m}} = \\
&= \frac{\left(\alpha-1\right)^{-2} \cdot \left(1+\frac{\alpha \cdot \eta_{\alpha,m}^2}{(\alpha-1)^2}\right)^m}{\left(1+\frac{\alpha \cdot \eta_{\alpha,m}^2}{(\alpha-1)^2}\right)^m} = \left(\alpha-1\right)^{-2}.
\end{align*}
Furthermore, we can expand ${}_2F_1(a,b,c,z)$ into Taylor series around $0$ in a following way:
$$
{}_2F_1\left(\frac{1}{2},1+m,\frac{3}{2},-\frac{\alpha\cdot t^2}{(\alpha-1)^2}\right) = 1-\frac{\alpha(m+1)t^2}{3(\alpha-1)^2} + O(t^4) \geqslant 1-\frac{\alpha \cdot (m+1) \cdot \eta_{\alpha,m}^2}{3\cdot\left(\alpha-1\right)^2}.
$$
Using these two observations we obtain
$$
G(\eta_{\alpha,m}) \geqslant \frac{4(\alpha-1)^{-2}\cdot \alpha \cdot m \cdot \eta_{\alpha,m} \cdot \left(1-\frac{\eta_{\alpha,m}^2}{2}\right) \cdot \left(1-\frac{\alpha \cdot (m+1) \cdot \eta_{\alpha,m}^2}{3\cdot\left(\alpha-1\right)^2}\right)}{\pi}
$$
We can further simplify this by recalling that $\alpha = e^{\frac{\epsilon}{\log_2(n)+1}}$ and $m = \gamma n$ and observing that $\left(1-\frac{\eta_{\alpha,m}^2}{2}\right) \cdot\left(1-\frac{\alpha \cdot (m+1) \cdot \eta_{\alpha,m}^2}{3\cdot\left(\alpha-1\right)^2}\right)$ is increasing with $n$. Let us call this value $c_n^*$. We can fix this for the smallest $n$ that we want to consider. See that, for example, for $n \geqslant 2^7$ we have $c_n^* \geqslant 1.43$. This leaves us with
\begin{align*}
G(\eta_{\alpha,m}) &\geqslant \frac{4 c_n^* \cdot (\alpha-1)^{-2}\cdot \alpha \cdot m \cdot \eta_{\alpha,m}}{\pi} = \\
&= \frac{4 c_n^* \cdot (\alpha-1)^{-2} \cdot \alpha \cdot m \cdot \sqrt{\frac{\pi(\alpha-1)^2}{4\alpha m}}}{\pi} = \\ 
&= \frac{2 c_n^* \cdot \sqrt{\alpha\cdot m}}{\sqrt{\pi} \cdot (\alpha-1)} \geqslant  \frac{2 c_n^* \cdot \sqrt{m}}{\sqrt{\pi} \cdot (\alpha-1)} = \frac{2c_n^* \cdot \sqrt{\gamma n}}{\sqrt{\pi} \cdot (e^{\frac{\epsilon}{\log_2(n)+1}}-1)} \geqslant \\
&\geqslant \frac{2c_n^* \cdot \sqrt{\gamma n}}{\sqrt{\pi} \cdot (e^{\frac{\epsilon}{\log_2(n)}}-1)} \geqslant \frac{\xi \log_2(n) \cdot 2c_n^* \cdot \sqrt{\gamma n}}{\epsilon \sqrt{\pi}},
\end{align*}
where $\xi$ is such that $e^{\xi\cdot x} \leqslant (1+x)$ for $x = (\frac{1}{2\log_2(n)})$. For example, in case we have $\epsilon = 0.5$ and $n \geqslant 2^7$ it suffices to take $\xi = 0.96$. In the end we have
$$
G(\eta_{\alpha,m}) \geqslant c_{n,\epsilon} \cdot \sqrt{\gamma} \cdot \frac{\log_2(n) \cdot \sqrt{n}}{\epsilon \sqrt{\pi}},
$$
where $c_{n,\epsilon} = 2\xi c_n^* $ which is, for moderate $n$ and $\epsilon$, greater than $1.4$. In fact, for $\epsilon = 0.5$ and $n \geqslant 2^7$ it is greater than $2$. In the end we have
$$
E|Z| \geqslant c_{n,\epsilon} \cdot \sqrt{\gamma} \cdot \frac{\log_2(n) \cdot \sqrt{n}}{\epsilon \sqrt{\pi}} - 0.1~,
$$
which completes the proof of this lemma.
\qed
\end{proof}

\subsection{Proof of Theorem~\ref{ssect:thmACC}}

\begin{proof}
Let us note that $\mathcal{S}$ is \textbf{not} connected if and only if there exists a subset of nodes from $\mathcal{S}$ with cardinality $1\leq k \leq m/2$ such that
there is no connection to any of remaining $m-k$ nodes. For a given subset of $\mathcal{S}$ of cardinality $k$ probability that no edge connects it to other $m-k$ nodes of $\mathcal{S}$ is  $(1-p)^{k(m-k)}$.

Let $A_k$ be an event that there exists such a "cut-off" subset of cardinality $k$.
Clearly, using union bound argument we get

$$\Pr[A_k]\leq (1-p)^{k(m-k)} {m \choose k}.$$ 
\noindent
Probability that $\mathcal{S}$ is not connected is equivalent to the event
$A_1\cup \ldots \cup A_k$ for $k=1,\ldots , m/2$. Again, using union bound

\begin{align*}
\Pr[A_1\cup \ldots \cup A_k] &  \leq \sum\limits_{i=1}^{m/2} \Pr[A_i] \leq \sum\limits_{i=1}^{m/2} (1-p)^{k(m-k)} {m \choose k} \leq \\
& \leq \sum\limits_{i=1}^{m/2} (1-p)^{k\frac{m}{2}} {m \choose k} = (\star).
\end{align*}

Since ${m \choose  k} \leqslant m^k$ we get

$$(\star) \leq \sum\limits_{i=1}^{m/2} \left( (1-p)^{\frac{m}{2}} m  \right)^{k} 
\leq \sum\limits_{i=1}^{\infty} \left( (1-p)^{\frac{m}{2}} m  \right)^{k} = \frac{(1-p)^{m/2}m}{1- (1-p)^{m/2}m} = (\star \star).$$

\noindent
Since the function $f(x)= \frac{a^x x}{ 1 - a^x x } $ is decreasing for $x> -\frac{1}{\log(a)}$ (if $0<a<1$) and from the assumption that $m \geq n/2$ we have

$$(\star \star) \leq \frac{(1-p)^{n/4}\frac{n}{2}}{1- (1-p)^{n/2}\frac{n}{2}}~.$$

Applying inequality $\exp(x) \geq 1+x$ and substituting $p=\frac{8\log n}{n}$ we  obtain
$$(\star \star) \leq \frac{\exp\left(-\frac{8\log(n)}{n}\right)\frac{n}{2}}{1- 1/2}=\exp\left(-\log(n^2)\right) n = \frac{1}{n}, $$
which concludes the proof of this theorem.
\qed

\end{proof}
\end{document}